\documentclass[11pt, a4paper]{article}

\usepackage{typearea} \paperwidth 8.1in \paperheight 11.5in \typearea{15}

\usepackage{mathrsfs}
\usepackage{color,epsfig}
\usepackage{amsmath,amsfonts,amssymb,amstext,amsthm}
\usepackage{xspace}
\usepackage{algorithm}
\usepackage[noend]{algorithmic}

\allowdisplaybreaks

\newtheorem{theorem}{Theorem}
\newtheorem{observation}{Observation}
\newtheorem{lemma}{Lemma}
\newtheorem{fact}{Fact}
\newtheorem{cl}{Claim}
\newtheorem*{claim*}{Claim}
\newtheorem{corollary}{Corollary}

\newcommand{\initOneLiners}{%
    \setlength{\itemsep}{0pt}
    \setlength{\parsep }{0pt}
    \setlength{\topsep }{0pt}
}

\newenvironment{proofof}[1]{

\bigskip\noindent{\bf Proof of {#1}:}}
{\hfill$\square$
}

\def\oG{\overline{G}}

\newcommand{\sse}{\subseteq}

\DeclareMathOperator*{\E}{\mathbb{E}}

\def\script#1{\mathcal{#1}}

\def\sM{\ensuremath{\script{M}}\xspace}
\def\sQ{\ensuremath{\script{Q}}\xspace}
\def\sI{\ensuremath{\script{I}}\xspace}
\def\sC{\ensuremath{\script{C}}\xspace}
\def\sV{\ensuremath{\script{V}}\xspace}
\def\sE{\ensuremath{\script{E}}\xspace}

\def\dbm{\textsc{DegMat}\xspace}
\def\rsp{\ensuremath{\mathsf{RSP}}\xspace}
\def\p{{\cal P}}

\def\dim{\mathrm{dim}}

\def\R{\mathbb{R}}

\title{Approximation-Friendly Discrepancy Rounding}
\author{Nikhil Bansal\thanks{Department of Mathematics and Computer Science, Eindhoven University of Technology. } \and Viswanath Nagarajan\thanks{Department of Industrial and Operations Engineering, University of Michigan.} }

\begin{document}
\maketitle

\begin{abstract}
Rounding linear programs using techniques from discrepancy is a recent approach that has been very successful in certain settings. However this method also has some limitations when compared to approaches such as randomized and iterative rounding.  
We provide an extension of the discrepancy-based rounding algorithm due to Lovett-Meka that (i) combines the advantages of both randomized and iterated rounding, (ii) makes it applicable to settings with more general combinatorial structure such as matroids. As applications of this approach, we obtain new results for various classical problems such as 
linear system rounding, degree-bounded matroid basis and low congestion routing.
\end{abstract}

\section{Introduction}
A very common approach for solving discrete optimization problems is to solve some linear programming relaxation, and then  round the fractional solution into an integral one, without (hopefully) incurring much loss in quality.
Over the years several ingenious rounding techniques have been developed (see e.g.~\cite{V01,SW11})
based on ideas from optimization, probability, geometry, algebra and various other areas. Randomized rounding and iterative rounding are two of the most commonly used methods.

Recently, discrepancy-based rounding approaches have also been very successful; a particularly notable result is for bin packing due to  Rothvoss~\cite{R13}. Discrepancy is a well-studied area in combinatorics with several surprising results (see e.g.~\cite{Matousek2010geometric}), and as observed by Lov\'{a}sz et al.~\cite{LSV86}, has a natural  
connection to rounding. However, until the recent algorithmic developments \cite{B10,LM12,HSS14,NT15,R-convex14}, most of the results in discrepancy were non-constructive and hence not directly useful for rounding.
These algorithmic approaches combine probabilistic approaches like randomized rounding with linear algebraic approaches such as iterated rounding \cite{Lau-book}, which makes them quite powerful. 

Interestingly, given the connection between discrepancy and rounding, these discrepancy algorithms 
can in fact be viewed as meta-algorithms for rounding. We discuss this in \S\ref{sec:prelim} in the context of the Lovett-Meka (LM) algorithm \cite{LM12}.
This suggests the possibility of one single approach that generalizes both randomized and iterated rounding.
This is our motivating goal in this paper.

While the LM algorithm is already an important step in this direction, it still has some important limitations.
For example, it is designed for obtaining additive error bounds and  
it does not give good multiplicative error bounds (like those given by randomized rounding). This is not an issue for discrepancy applications, but crucial for many approximation algorithms.
Similarly, iterated rounding can work well with exponentially sized LPs by exploiting their underlying  combinatorial structure (e.g., degree-bounded spanning tree~\cite{SL07}), but the current discrepancy results \cite{LM12,R-convex14} give extremely weak bounds in such settings.

\paragraph{Our Results:}
We extend the LM algorithm to overcome the limitations stated above. 
In particular, we give a new variant that also gives Chernoff type multiplicative error bounds 
(sometimes with an additional logarithmic factor loss).
We also show how to adapt the above algorithm to handle exponentially large LPs  involving matroid constraints,  as in iterated rounding.

This new discrepancy-based algorithm gives new results for problems such as linear system rounding with violations~\cite{BF81,LLRS01}, degree-bounded matroid basis~\cite{KLS08,CVZ10}, low congestion routing~\cite{KLRTVV87,LLRS01} and multi-budgeted matroid basis~\cite{GRSZ14}, 
These results simultaneously combine non-trivial guarantees from discrepancy, randomized rounding and iterated rounding and previously such bounds were not even known existentially.

Our results are described formally in \S\ref{sec:result}. To place them in the proper context, we first need to describe some existing rounding approaches (\S\ref{sec:prelim}). The reader familiar with the LM algorithm can directly go to \S\ref{sec:result}.

\subsection{Preliminaries}\label{sec:prelim}
We begin by describing LM rounding \cite{LM12}, randomized rounding and iterated rounding in a similar form, and then discuss their strengths and weaknesses.

 \paragraph{LM Rounding:} Let $A$ be a $m\times n$ matrix with $0-1$ entries\footnote{The results below generalize to arbitrary real matrices $A$ and vectors $x$ in natural ways, but we consider $0$-$1$ case for simplicity.}, $x \in [0,1]^n$ a fractional vector and let $b=Ax$. Lovett and Meka~\cite{LM12} showed the following rounding result.

\begin{theorem} [LM Rounding~\cite{LM12}]
\label{th:lm}
Given $A$ and $x$ as above,
For $j=1,\ldots,m$, pick any $\lambda_j$ satisfying
\begin{equation}
\label{eqcond21}
 \sum_j \exp(-\lambda_j^2/4) \leq n/16.
 \end{equation}
Then there is an efficient randomized algorithm to find a solution $x'$ such that: (i) at most $n/2$ variables of $x'$ are fractional (strictly between  $0$ and $1$) and, 
(ii) $|\langle a_j , x'- x\rangle| \leq \lambda_j \|a_j\|_2 $ for each $j=1,\ldots,m$, where  $a_j$ denotes the $j$-th row of $A$.
\end{theorem}

\noindent {\em Remark:} 
The right hand side of \eqref{eqcond21} can be set to $(1-\epsilon )n$ for any fixed constant $\epsilon>0$, at the expense of 
$O_\epsilon(1)$  factor loss in other parameters of the theorem; see e.g.~\cite{BCKL14}.

\paragraph{Randomized Rounding:} Chernoff bounds state that if $X_1,\ldots,X_n$ are independent Bernoulli random variables, and $X=\sum_i X_i$ and $\mu = \mathbb{E}[X]$, then 
\[\mathrm{Pr}[ |X - \mu|  \geq \epsilon \mu] \leq 2 \exp(-\epsilon^2 \mu/4)  \qquad \mbox{for } \epsilon \leq 1. \]
Then independent randomized rounding can be viewed as the following (by using Chernoff bounds and union bound, and denoting $\lambda_j = \epsilon_j \sqrt{b_j}$). 

\begin{theorem}[Randomized Rounding]
\label{th:chernoff}
For $j=1,\ldots,m$, pick any $\lambda_j$ satisfying $\lambda_j \leq \sqrt{b_j}$, and
\begin{equation}
\label{eqcond31}
 \sum_j \exp(-\lambda_j^2/4) < 0.5
 \end{equation}
Then independent randomized rounding gives a solution $x'$ such that:
(i) All variables are $0$-$1$, and (ii) $|\langle a_j , x'- x\rangle| \leq \lambda_j \sqrt{b_j}$ for each $j=1,\ldots,m$.
\end{theorem}

\paragraph{Iterated Rounding~\cite{Lau-book}:} 
This is based on the following linear-algebraic fact.
\begin{theorem}
\label{th:iter}
If $m<n$, then there is a solution $x' \in [0,1]^n$ such that  (i) $x'$ has at least $n-m$ variables set to $0$ or $1$ and, (ii) $A(x'-x) = 0$  (i.e.,~$b=Ax'$).
\end{theorem}
In iterated rounding applications, if $m>n$  then some cleverly chosen constraints are dropped until $m<n$ and  integral variables are obtained. This is done repeatedly.

\paragraph{Strengths of LM rounding:} 
Note that if we set $\lambda_j \in \{0,\infty\}$ in LM rounding, then
it gives a very similar statement to Theorem \ref{th:iter}.
E.g., if we only care about some $m=n/2$ constraints  
then Theorem \ref{th:iter} gives an $x'$ with at least $n/2$ integral variables and $a_jx=a_jx'$ for all these $m$ constraints.
Theorem \ref{th:lm} (and the remark below it) give the same guarantee if we set $\lambda_j=0$ for all constraints.  In general, LM rounding can be much more flexible as it allows arbitrary $\lambda_j$.

Second, LM rounding is also related to randomized rounding.
Note that \eqref{eqcond31} and \eqref{eqcond21} have the same left-hand-side. However, the right-hand-side of \eqref{eqcond21} is $\Omega(n)$, while that of \eqref{eqcond31} is $O(1)$.
This actually makes a huge difference. In particular, in \eqref{eqcond31} one cannot set $\lambda_j=1$ for more than a couple of constraints (to get an $o(\sqrt{b_j})$ error bound on constraints), while in \eqref{eqcond21}, one can even set $\lambda_j=0$ for $O(n)$ constraints. In fact, almost all non-trivial results in discrepancy \cite{Spencer85,Srin97,Matousek2010geometric} are based on this ability.

\paragraph{Weaknesses of LM rounding:}  First, Theorem~\ref{th:lm} only gives a partially integral solution instead of a fully integral one as in Theorem \ref{th:chernoff}. 

Second, and more importantly, it only gives additive error bounds instead of multiplicative ones. In particular, note the $\lambda_j \|a_j\|_2$ vs $\lambda_j \sqrt{b_j}$ error in Theorems \ref{th:lm} and \ref{th:chernoff}.
E.g.,~for a constraint 
$\sum_i x_i = \log n$, Theorem \ref{th:chernoff} gives $\lambda \sqrt{\log n}$ error but Theorem \ref{th:lm} gives a much higher $\lambda \sqrt{n}$ error. 
So, while randomized rounding can give a good multiplicative error like $a_jx' \le (1 \pm \epsilon_j)b_j$, LM rounding is completely insensitive to $b_j$.

Finally, iterated rounding works extremely well in many settings where Theorem \ref{th:lm} does not give anything useful.
E.g.,~in problems involving exponentially many constraints such as the degree bounded spanning tree problem. The problem is that if
$m$ is exponentially large,
then the $\lambda_j$'s in Theorem \ref{th:lm} need to be very large to satisfy \eqref{eqcond31}.

\subsection{Our Results and techniques}\label{sec:result}

Our first  result is the following improvement over Theorem~\ref{th:lm}.

\begin{theorem}\label{thm:main}
There is a constant $K_0>0$ and randomized polynomial time algorithm that given any $n>K_0$, fractional solution $y \in [0,1]^n$, $m\le 2^n$ linear constraints $a_1,\ldots,a_m \in \R^n$  and  $\lambda_1,\cdots,\lambda_m\ge 0$ with $\sum_{j=1}^m e^{-\lambda_j^2/K_0} <\frac{n}{16}$,  finds a solution $y'\in [0,1]^n$ such that:
\begin{eqnarray}
|\langle y'-y, a_j\rangle | & \le & \lambda_j\cdot \sqrt{W_j(y)} + \frac{1 }{n^2} \cdot \|a_j\|, \quad \forall j=1,\cdots m  \label{eq:main-violation}\\
y'_i\in\{0,1\}, &&\mbox{ for $\Omega(n)$ indices }i\in \{1,\cdots, n\} \label{eq:main-integral}
\end{eqnarray}
Here $W_j(y):=\sum_{i=1}^n a_{ji}^2\cdot \min\{y_i,1-y_i\}^2$ for each $j= 1,\cdots m$. 
\end{theorem}

\noindent {\em Remarks:} 1) The error $\lambda_j\sqrt{W_j(y)}$ is always smaller than $\lambda_j\|a_j\|$ in LM-rounding and $\lambda_j (\sum_{i=1}^n a_{ji}^2\cdot y_i(1-y_i))^{1/2}$ in randomized rounding. In fact it could even be much less  if the $y_{i}$ are very close to $0$ or $1$.

\noindent 2) The term $n/16$ above can be made  $(1-\epsilon)n$ for any fixed constant $\epsilon>0$, at the expense of worsening other constants (just as in LM rounding).

\noindent 3) The additional error term $\frac{1}{n^2} \cdot \|a_j\|$ above is negligible and can be reduced to $\frac{1}{n^c} \cdot \|a_j\|$ for any constant $c$, at the expense of a larger running time $n^{O(c)}$.

We note that Theorem~\ref{thm:main} can also be obtained in a ``black box'' manner from LM-rounding (Theorem~\ref{th:lm}) by rescaling the polytope and using its symmetry.\footnote{We thank an anonymous reviewer for pointing this out.} However, such an approach does not work in the setting of matroid polytopes (Theorem~\ref{thm:mat-disc}). In the matroid case, we need to modify LM-rounding as outlined below. 

\medskip 

\noindent{\bf Applications:}  We focus on  {\em linear system rounding} as the prime example. Here,
given matrix $A\in [0,1]^{m\times n}$ and vector $b\in \mathbb{Z}_+^m$, the goal is to find a vector $z\in \{0,1\}^n$ satisfying $Az = b$. As this is NP-hard, the focus has been on finding a $z\in \{0,1\}^n$ where $Az \approx b$.

Given any fractional 
solution $y\in[0,1]^n$ satisfying $Ay = b$, using Theorem~\ref{thm:main} iteratively we can obtain an integral vector $z\in \{0,1\}^n$ with
\begin{equation}\label{eq:set-pack-bound}
|a_j z -b_j|\le \min\left\{ O(\sqrt{n \log(2+m/n)})\,,\,\, \sqrt{L \cdot b_j} + L \right\},\quad \forall j\in[m],
\end{equation}
where  $L=O(\log n \log m)$ and $[m]:=\{1,2,\cdots m\}$.\footnote{For any integer $t\ge 1$, we use the notation $[t]:=\{1,2,\cdots,t\}$.}  
Previously known algorithms could provide a bound of either $O(\sqrt{n \log(m/n)})$ for all constraints~\cite{LM12} 
  or $O(\sqrt{\log m}\cdot \sqrt{b_j} + \log m)$ for all constraints (Theorem~\ref{th:chernoff}). 
 Note that this does not imply a $\min\{ \sqrt{n \log(m/n)}, \sqrt{\log m}\cdot \sqrt{b_j} + \log m\}$ violation per constraint,
as in general it is not possible to combine two integral solutions and achieve the better of their violation bounds on all constraints. To the best of our knowledge, even the existence of an integral solution satisfying the bounds in~\eqref{eq:set-pack-bound} was not known prior to our work.

In the setting where the matrix $A$ is ``column sparse'', i.e. each variable appears in at most $\Delta$ constraints, we obtain a more refined error of 
\begin{equation}\label{eq:set-pack-bound2}
|a_j y -b_j |\le \min\left\{ O(\sqrt{\Delta}\log n)\,,\,\, \sqrt{L \cdot b_j} + L \right\},\quad \forall j\in [m],
\end{equation}
 where $L=O(\log n\cdot \log m)$. Previous algorithms could separately achieve bounds of $\Delta-1$~\cite{BF81}, $O(\sqrt{\Delta}\log n)$~\cite{LM12} or 
$O(\sqrt{\log \Delta}\cdot \sqrt{b_j} + \log \Delta)$~\cite{LLRS01}. For clarity, Figure~\ref{fig:sp-graph} plots the violation bounds achieved by these different algorithms as a function of the right-hand-side $b$ when $m=n$ (we assume $b,\Delta\ge \log^2 n$). Note again that since there are multiple constraints 
we can not simply combine algorithms to achieve the smaller of their violation bounds.

\begin{figure}
\includegraphics[scale=0.85]{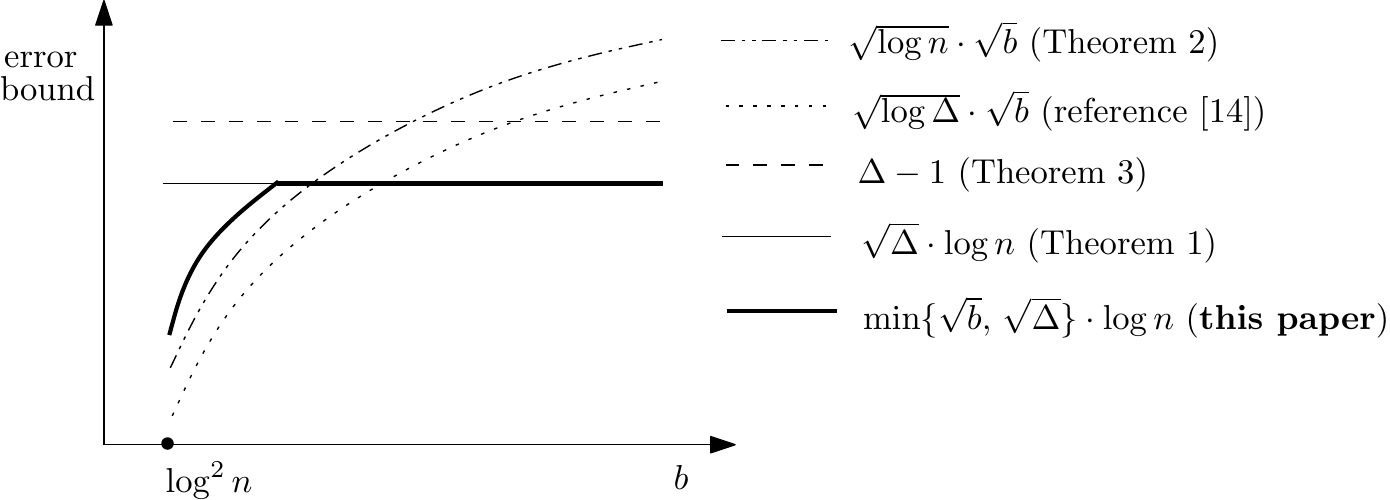}
\caption{Additive violation bounds for linear system rounding when $\Delta\ge \log^2n$ and $b\ge \log^2n$.\label{fig:sp-graph}}
\end{figure}

One can also combine the bounds in~\eqref{eq:set-pack-bound} and~\eqref{eq:set-pack-bound2}, and use some additional ideas from discrepancy to obtain:
 \begin{equation}\label{eq:set-pack-bound3}
|a_j y -b_j|\, \le \, O(1)\cdot \min\left\{ \sqrt{j}, \,  \sqrt{n \log(2+\frac{m}{n})}, \, \sqrt{L \cdot b_j} + L, \,  \sqrt{\Delta}\log n \right\}, \quad  \forall j \in [m].
\end{equation}

\noindent{\bf Matroid Polytopes:} Our main result is an extension of Theorem~\ref{thm:main} where the fractional solution lies in a matroid polytope in addition to satisfying the linear constraints $\{a_j\}_{j=1}^m$.  Recall that a matroid \sM is a tuple $(V,\sI)$ where $V$ is the groundset of elements and $\sI\sse 2^V$ is a collection of independent sets satisfying the hereditary and exchange properties~\cite{Sch}. The rank function $r:2^V\rightarrow \mathbb{Z}$ of a matroid is defined as $r(S)=\max_{I\in \sI, I\sse S}\, |I|$. The matroid polytope (i.e. convex hull of all independent sets) is given by the following linear inequalities:

\[P(\sM)\quad := \quad \left\{ x\in \R^n\,\, :\,\, \sum_{i\in S}  x_i \le r(S) \,\, \forall S\sse V,\,\, \, x\ge 0\right\}.\]

As is usual when dealing with matroids, we assume access to an `` independent set oracle'' for \sM that given any subset $S\sse V$ returns whether/not $S\in \sI$ in polynomial  time. 

\begin{theorem}\label{thm:mat-disc}
There is  a randomized polynomial time algorithm that given  matroid $\sM$, fractional solution $y\in P(\sM)$, linear constraints $\{a_j \in \mathbb{R}^n\}_{j=1}^m$ 
and values $\{\lambda_j\}_{j=1}^m$ satisfying the conditions in Theorem~\ref{thm:main}, finds a solution $y'\in P(\sM)$ satisfying~\eqref{eq:main-violation}-\eqref{eq:main-integral}. 
\end{theorem}
We note that the same result can be obtained even if we want to compute a base (maximal independent set) in the matroid: the only difference here is to add the equality $\sum_{i\in V} x_i = r(V)$ to $P(\sM)$ which corresponds to the base polytope of \sM.   

The fact that we can {\em exactly} preserve  matroid constraints leads to a number of  improvements:

 \smallskip
\noindent {\em Degree-bounded matroid basis} (\dbm). Given a matroid on elements $[n]:=\{1,2,\cdots n\}$ with costs $d:[n]\rightarrow \mathbb{Z}_+$ and $m$ ``degree constraints'' $\{S_j, b_j\}_{j=1}^m$ where each $S_j\sse [n]$ and $b_j\in \mathbb{Z}_+$, the goal is to find a minimum-cost basis $I$ in the matroid that satisfies $|I\cap S_j|\le b_j$ for all $j\in[m]$. Since even the feasibility problem is NP-hard, we consider {\em bicriteria} approximation algorithms that violate the degree bounds. We obtain an algorithm where the solution costs at most the optimal and the degree bound violation is as in~\eqref{eq:set-pack-bound3}; here $\Delta$ denotes the maximum number of sets $\{S_j\}_{j=1}^m$ containing any element.

Previous algorithms achieved approximation ratios of  $(1,b+O(\sqrt{b\log n}))$~\cite{CVZ10}, based on randomized swap rounding,  and $(1,b+\Delta-1)$~\cite{KLS08} based on iterated rounding. 
Again, these bounds could not be combined together as they used different algorithms. 
We note that in general the $(1,b+O(\sqrt{n\log(m/n)}))$ approximation is the best possible (unless P=NP) for this problem~\cite{CNN11,BKKNP13}.

\smallskip\noindent {\em Multi-criteria matroid basis}.  Given a matroid on elements $[n]$ with $k$ different cost functions $d_i:[n]\rightarrow \mathbb{Z}_+$ (for $i=1,\cdots,k$) and budgets $\{B_i\}_{i=1}^k$, the goal is to find (if possible) a  basis  $I$ with $d_i(I)\le B_i$ for each $i\in [k]$. We obtain an algorithm that for any $\epsilon>0$ finds in $n^{O(k^{1.5}\,/\,\epsilon)}$ time, a basis $I$ with $d_i(I)\le (1+\epsilon)B_i$ for all $i\in[k]$. Previously, \cite{GRSZ14} obtained such an algorithm with $n^{O(k^2\,/\,\epsilon)}$ running time.

\smallskip\noindent {\em Low congestion routing}. Given a directed graph $G=(V,E)$ with edge capacities $b:E\rightarrow \mathbb{Z}_+$,  $k$ source-sink pairs $\{(s_i,t_i)\}_{i=1}^k$ and a length bound $\Delta$, the goal is to find an $s_i-t_i$ path $P_i$ of length at most $\Delta$ for each pair $i\in[k]$ such that the number $N_e$ of paths using any edge $e$ is at most $b_e$. Using an LP-based reduction~\cite{CVZ10} this can be cast as an instance of \dbm. So we obtain violation bounds as in~\eqref{eq:set-pack-bound3} which implies:
\[N_e \,\,\le \,\, b_e+\min\left\{O(\sqrt{\Delta }\log n),\, O(\sqrt{b_e} \log n+\log^2n)\right\},\quad \forall e\in E.\] 
Here $n=|V|$ is the number of vertices. Previous algorithms achieved  bounds of $\Delta-1$~\cite{KLRTVV87} or $O(\sqrt{\log \Delta}\cdot \sqrt{b_j} + \log \Delta)$~\cite{LLRS01} separately. We can also handle a richer set of routing requirements: given a laminar family ${\cal L}$ on the $k$ pairs, with a requirement $r_T$ on each set $T\in {\cal L}$, we want to find a multiset of paths so that there are at least $r_T$ paths between the pairs in each $T\in {\cal L}$. Although this is not an instance of \dbm, the same approach works.

\medskip

\noindent {\bf Overview of techniques:} Our algorithm in Theorem \ref{thm:main} is similar to the Lovett-Meka algorithm, and is also based on performing a Gaussian random walk at each step in a suitably chosen subspace.
However, there some crucial differences. First, instead of updating each variable by the standard Gaussian $N(0,1)$, the variance for variable $i$ is chosen proportional to $\min(y_i,1-y_i)$, i.e.~proportional to how close it is to the boundary $0$ or $1$. This is crucial for getting the multiplicative error instead of the additive error in the constraints.
However, this slows down the ``progress" of variables toward reaching $0$ or $1$. To get around this, we add $O(\log n)$ additional constraints to define the subspace where the walk is performed: these restrict the total fractional value of variables in a particular ``scale'' to remain fixed. Using these we can  ensure that enough variables eventually reach $0$ or $1$. 

In order to handle the matroid constraints (Theorem~\ref{thm:mat-disc}) we need to incorporate them (although they are exponentially many) in defining the subspace where the random walk is performed. One difficulty that arises here is that we can no longer implement the random walk using ``near tight'' constraints as in~\cite{LM12} since we are unable to bound the dimension of near-tight matroid constraints.  However, as is well known, the dimension of {\em exactly} tight matroid constraints is at most $n/2$ at any (strictly) fractional solution, and so we implement the random walk using exactly tight constraints. This requires us to truncate certain steps in the random walk (when we move out of the polytope), but we show that the effect of such truncations is negligible.


\section{Matroid Partial Rounding}\label{sec:main-rounding}
In this section we will prove Theorem \ref{thm:mat-disc} which also implies Theorem \ref{thm:main}.

We may assume, without loss of generality, that $\max_{j=1}^m\lambda_j\le n$. This is because setting $\mu_j=\min\{\lambda_j , n\}$ we have $\sum_{j=1}^m e^{-\mu_j^2/K_0} \le \sum_{j=1}^m e^{-\lambda_j^2/K_0}  + m\cdot e^{-n} < \frac{n}{16}+1$ (we used the assumption $m\le 2^n$). So we can apply  Theorem \ref{thm:mat-disc} with $\mu_j$s instead of $\lambda_j$s to obtain a stronger result.

Let $y \in \R^n$ denote the initial solution. The algorithm will start with $X_0=y$ and update this vector over time. Let $X_t$ denote the vector at time $t$ for $t=1,\ldots,T$. The value of $T$ will be defined later.
Let $\ell= 3\lceil \log_2 n\rceil$.
We classify the $n$ elements into $2\ell$ classes based on their initial values $y(i)$ as follows.

$$U_k:=\left\{ 
\begin{array}{ll}
\left\{i\in[n] : 2^{-k-1} < y(i) \le 2^{-k} \right\} & \mbox{ if } 1\le k\le \ell-1\\
\left\{i\in[n] : y(i) \le 2^{-\ell} \right\} & \mbox{ if } k = \ell.
\end{array}\right.$$

$$V_k:=\left\{ 
\begin{array}{ll}
\left\{i\in[n] : 2^{-k-1} < 1-y(i) \le 2^{-k} \right\} & \mbox{ if } 1\le k\le \ell-1\\
\left\{i\in[n] : 1-y(i) \le 2^{-\ell} \right\} & \mbox{ if } k = \ell.
\end{array}\right.$$

Note that the $U_k$'s partition elements of value (in $y$) between $0$ and $\frac{1}{2}$ and the $V_k$'s form a symmetric partition of elements valued between $\frac{1}{2}$ and $1$. 
This partition does not change over time, even though the value of variables might change.
We define the ``scale'' of each element as:
$$s_i \,\,:=\,\, 2^{-k}, \qquad \forall i\in U_k\cup V_k,\quad \forall k\in [\ell].$$

Define $W_j(s)= \sum_{i=1}^n a_{ji}^2\cdot s_i^2$ for each $j\in[m]$. Note that $W_j(s)\ge W_j(y)$ and
$$W_j(s) - 4\cdot W_j(y) \le \sum_{i=1}^n a_{ji}^2\cdot \frac{1}{n^6} = \frac{\|a_j\|^2}{n^6}.$$
So $\sqrt{W_j(y)} \le \sqrt{W_j(s)} \le 2\sqrt{W_j(y)} +\frac{\|a_j\|}{n^3}$. Our algorithm will find a solution $y'$ with $\Omega(n)$ integral variables such that:
$$|\langle y'-y, a_j\rangle | \le \lambda_j\cdot \sqrt{W_j(s)} + \frac{1}{n^3} \cdot \|a_j\|  , \quad \forall j\in[m].$$

This suffices to prove Theorem~\ref{thm:mat-disc} as  
$$\lambda_j\cdot \sqrt{W_j(s)} + \frac{1}{n^3} \cdot \|a_j\|   \le 2\lambda_j\cdot \sqrt{W_j(y)} + \left( \frac{1}{n^3}+\frac{\lambda_j}{n^3}\right) \cdot \|a_j\| \le 2\lambda_j\cdot \sqrt{W_j(y)} + \frac{\|a_j\| }{n^2}.$$

Consider the polytope $\sQ$ of points $x \in \R^n$ satisfying the following constraints.
\begin{eqnarray}
x & \in &  P(\sM), \label{eq:Q-mat}\\
 |\langle x-y, a_j\rangle | & \le &  \lambda_j\cdot \sqrt{W_j(s)} +\frac1{n^3}\cdot \|a_j\| \qquad \forall j\in[m], \label{eq:Q-deg}\\
 \sum_{i\in U_k} x_i  & = & \sum_{i\in U_k} y_i \qquad  \qquad  \qquad  \qquad \qquad \forall k \in [\ell], \label{eq:Q-part1}\\
 \sum_{i\in V_k} x_i  & = & \sum_{i\in V_k} y_i  \qquad  \qquad  \qquad \qquad  \qquad \forall k \in [\ell], \label{eq:Q-part2}\\
 0 \,\,\, \le  \,\,\, x_i & \le & \min\{ \alpha \cdot 2^{-k}, 1\}   \qquad  \qquad   \qquad \forall i\in U_k, \, \forall k\in [\ell], \label{eq:Q-var1}\\
 0\,\,\,  \le \,\,\, 1- x_i & \le&  \min\{ \alpha \cdot 2^{-k}, 1\}\qquad  \qquad \qquad   \forall i\in V_k, \, \forall k\in [\ell].\label{eq:Q-var2}
\end{eqnarray}

Above $\alpha=40$ is a constant whose choice will be clear later. 
The algorithm will maintain the invariant that at any time $t \in [T]$, the solution $X_t$ lies in $\sQ$.
In particular the constraint \eqref{eq:Q-mat} requires that $X_t$ stays in the matroid polytope.
Constraint \eqref{eq:Q-deg} controls the violation of the side  constraints over all time steps.
The last two constraints \eqref{eq:Q-var1}-\eqref{eq:Q-var2} enforce that variables in $U_k$ (and symmetrically $V_k$) do not deviate far beyond their original scale of $2^{-k}$. The constraints \eqref{eq:Q-part1} and \eqref{eq:Q-part2} 
ensure that throughout the algorithm, the total value of elements in $U_k$ (and $V_k)$ stay equal to their initial sum (in $y$).
These constraints will play a crucial role in arguing that the algorithm finds a partial coloring.
Note that there are only $2 \ell$ such constraints.

In order to deal with complexity issues, we will assume (without loss of generality,  by scaling) that all entries in the constraints describing $\sQ$ are integers bounded by some value $B$. 
Our algorithm will then run in time polynomial in $n,m$ and $\log_2 B$, given  an independent set oracle for the matroid $\sM$. 
Also, our algorithm will only deal with points having rational entries of small ``size''. Recall that the size of a rational number is the number of bits needed to represent it, i.e.  the size of  $p/q$ (where $p,q\in \mathbb{Z}$) is $\log_2 |p| + \log_2|q|$.

\medskip

\noindent {\bf The Algorithm:}
Let $\gamma=n^{-6}$ and $T=K/\gamma^2$ where $K:=10\alpha^2$. The algorithm starts with solution $X_0 = y \in \sQ$, and does the following at each time step $t=0,1,\cdots, T$:
\begin{enumerate}
\item Consider the set of constraints of $\sQ$ that are tight at the point $x=X_t$, and define the following sets based on this. 
\begin{enumerate}
\item Let $\sC^{var}_t$ be the set of tight variable constraints among~\eqref{eq:Q-var1}-\eqref{eq:Q-var2}. This consists of:
\begin{enumerate}
\item $i\in U_k$ (for any $k$) with $X_t(i)=0$ or $X_t(i)=\min\{\alpha \cdot 2^{-k},\, 1\}$; and  
\item $i\in V_k$ (for any $k$) with $X_t(i)=1$ or $X_t(i)=\max\{1-\alpha \cdot 2^{-k},\, 0\}$.
\end{enumerate}
 \item Let $\sC^{side}_t$ be the set of tight side constraints from~\eqref{eq:Q-deg}, i.e. those $j\in[m]$ with $$|\langle X^t - y, a_j\rangle | =\lambda_j\cdot \sqrt{W_j(s)}+\frac1{n^3}\, \|a_j\|.$$
\item Let $\sC^{part}_t$ denote the set of the $2\ell$ equality constraints~\eqref{eq:Q-part1}-\eqref{eq:Q-part2}.  
\item\label{step:rank-tight} Let $\sC^{rank}_t$ be a maximal linearly independent set of tight rank constraints  
from~\eqref{eq:Q-mat}. As usual, a set of constraints is said to be linearly independent if the corresponding coefficient vectors are linearly independent. Since $\sC^{rank}_t$ is maximal, every tight rank constraint is a linear combination of  constraints in $\sC^{rank}_t$. 
By Claim~\ref{cl:mat-tight-rank},  $|\sC^{rank}_t| \leq n/2$.
\end{enumerate}
\item Let $\sV_t$ denote the subspace orthogonal to all the constraints in $\sC^{var}_t$, $\sC^{side}_t$, $\sC^{part}_t$ and $\sC^{rank}_t$.
Let $D$ be an $n\times n$ diagonal matrix with entries $d_{ii} = 1/s_i$, and let $\sV'_t$ be the subspace
$\sV'_t = \{Dv : v \in \sV_t\}.$
As $D$ is invertible, $\dim(\sV'_t)  = \dim(\sV_t)$.
 
\item \label{step:orthn} Let $\{b_1,\ldots,b_k\}$ be an  almost orthonormal  basis of $\sV'_t$ given by Fact~\ref{fact:gauss-elim}. Note that all entries in these vectors are rationals of size $O(n^2\log B)$.

\item\label{step:gaussian} Let $G_t$ be a random direction defined as  
$G_t := \sum_{h=1}^k g_h b_h$ where the $g_h$ are independent $\{-1,+1\}$ Bernoulli random variables. 

\item Let $\oG_t := D^{-1} G_t$.  As $G_t \in \sV'_t$, it must be that $G_t = D v$ for some $v \in \sV_t$ and thus $\oG_t = D^{-1}G_t \in \sV_t$. Note that all entries in $\oG_t$ are rationals of size $O(n^3\log B)$. 

\item Set $Y_t=X_t + \gamma\cdot \oG_t$.
\begin{enumerate}
\item  If $Y_t\in \sQ$ then $X_{t+1}\gets Y_t$ and continue to the next iteration.
\item \label{step:walk-trunc} Else $X_{t+1}\gets $ the point in $\sQ$ that lies on the line segment $(X_t,Y_t)$ and is closest to $Y_t$. This can be found by binary search and testing membership in the matroid polytope. By Claim~\ref{cl:bin-search} it follows that the number of steps in the binary search is at most $O(n\log B)$.
\end{enumerate} 
\end{enumerate}

This completes the description of the algorithm. We actually do not need to compute the tight constraints %
from scratch in each iteration. We start the algorithm off with a strictly feasible solution $y\in \sQ$ which does not have any tight constraint other than~\eqref{eq:Q-part1}-\eqref{eq:Q-part2}. Then, the only place a new constraint gets tight is Step~\ref{step:walk-trunc}: at this point, we add the new constraint to the appropriate set among $\sC^{var}_t$, $\sC^{side}_t$  and $\sC^{var}_t$ and continue.

In order to keep the analysis clean and convey the main ideas, we will assume that the basis $\{b_1,\cdots b_k\}$ in Step~\ref{step:orthn} is exactly orthonormal. When the basis is ``almost orthonormal'' as given in Fact~\ref{fact:gauss-elim}, the additional error incurred is negligible.

\paragraph{Running Time.} Since the  number of iterations   is polynomial, we only need to show that each of the steps in any single iteration can be implemented in polynomial time.  The only step that requires justification is  \ref{step:walk-trunc}, which is shown in Claim~\ref{cl:bin-search}. Moreover, we need to ensure that all points considered in the algorithm have rational coefficients of polynomial size. This is done by a rounding procedure (see Fact~\ref{fact:LP-point}) that given an arbitrary point, finds a nearby rational point of size $O(n^2\log B)$. Since the number of steps in the algorithm is polynomial, the total error incurred by such rounding steps is small.

\begin{cl}\label{cl:bin-search}
The number of binary search iterations performed in Step~\ref{step:walk-trunc} is $O(n^4\log B)$.
\end{cl}
\begin{proof}
To reduce notation let $a=X_t$, $d=\gamma \oG_t$ and  $Y(\mu) :=a + \mu\cdot d\in \R^n$ for any $\mu\in \R$.  Recall that Step~\ref{step:walk-trunc} involves finding the maximum value of $\mu$ such that point $Y(\mu) \in \sQ$.

By the rounding procedure (Fact~\ref{fact:LP-point}) we know that $a$ has rational entries of size $O(n^2\log B)$.

We now show that the direction $d$ has rational entries of size $O(n^3\log B)$. This is because (i) the basis vectors $\{b_1,\cdots ,b_k\}$ have rational entries of size $O(n^2\log B)$ by Fact~\ref{fact:gauss-elim}, (ii) $G_t=\sum_{h=1}^k g_h\cdot b_h$ (where each $g_h=\pm 1$)  has rational entries of size $O(n^3\log B)$ and (iii) 
$\oG_t=D^{-1}G_t$ where $D^{-1}$ is a diagonal matrix with rational entries of size $O(n\log B)$.

Next, observe  that for any constraint $\langle a',x\rangle \le \beta$ in $\sQ$, the point of intersection of the hyperplane $\langle a',x\rangle =\beta$ with  line  $\{Y(\mu) : \mu\in \R\}$ is $\mu=\frac{\beta-\langle a',a\rangle}{\langle a', d\rangle}$ which is a rational of size at most $\sigma = O(n^4\log B)$ as $a',a,d,\beta$ all have rational entries of size $O(n^3\log B)$. Let $\epsilon = 2^{-2\sigma}$  be a value such that the difference between any two distinct rationals of size  at most $\sigma$ is more than $\epsilon$. 

In Step~\ref{step:walk-trunc}, we start the  binary search  with the interval $[0,1]$ for $\mu$ where $Y(0)\in \sQ$ and $Y(1)\not\in \sQ$. We perform this binary search until the interval width falls below $\epsilon$, which requires $\log_2 \frac1\epsilon=  O(n^4\log B)$ iterations. At the end, we have two values $\mu_0<\mu_1$ with $\mu_1-\mu_0<\epsilon$ such that $Y(\mu_0)\in \sQ$ and $Y(\mu_1)\not\in \sQ$. Moreover, we obtain a constraint $\langle a',x\rangle \le \beta$ in $\sQ$ that is not satisfied by $Y(\mu_1)$. We set $\mu'$ to be the (unique)  value such that $Y(\mu')$ satisfies this constraint at equality, and set $X_{t+1}=Y(\mu')$. Note that $\mu_0\le \mu' <\mu_1$. To see that $Y(\mu')\in \sQ$, suppose (for contradiction) that some constraint in $\sQ$ is not satisfied at $Y(\mu')$; then the point of intersection of line $\{Y(\mu) : \mu\in \R\}$ with this constraint must be at $\mu\in [\mu_0, \mu')$ which (by the choice of $\epsilon$) can not be a rational of size at most $\sigma$--- a contradiction.   
\end{proof}

\medskip
\noindent {\bf Analysis.} The analysis involves proving the following main lemma.  
\begin{lemma}\label{lem:main}
With constant probability, the final solution $X_T$ has $|\sC^{var}_T|\ge \frac{n}{20}$.
\end{lemma}

We first show how this implies Theorem~\ref{thm:mat-disc}. 

\begin{proofof}{Theorem ~\ref{thm:mat-disc} from Lemma~\ref{lem:main}}
The algorithm outputs the solution $y':=X_T$. By design the algorithm ensures that $X_T\in \sQ$, and thus 
$X_T \in  P(\sM)$ and it satisfies the error bounds \eqref{eq:Q-deg} on the side constraints.
  It remains to show that $\Omega(n)$ variables in $X_T$ must be integer valued  whenever $|\sC^{var}_T|\ge \frac{n}{20}$. For each $k\in [\ell]$ define 
$u_k := |\{i\in U_k : X_T(i)=\alpha \cdot 2^{-k}\}|$
and $v_k := |\{i\in V_k : X_T(i)=1-\alpha \cdot 2^{-k}\}|$. 
By the equality constraints \eqref{eq:Q-part1} for $U_k$, it follows that 
$$u_k\cdot \alpha \cdot 2^{-k} \le \sum_{i\in U_k} X_T(i)  = X_T(U_k) = y(U_k) \le |U_k|\cdot 2^{-k}.$$
This gives that $u_k\le \frac{1}{\alpha}|U_k|$. Similarly, $v_k\le \frac{1}{\alpha}|V_k|$. This implies that  
 $\sum_{k=1}^\ell (u_k+v_k)\le n/\alpha$. As the tight variables in 
 $\sC^{var}_t$ have values either $0$ or $1$ or $\alpha \cdot 2^{-k}$ or $1-\alpha\cdot 2^{-k}$, it follows that the number of $\{0,1\}$ variables is at least $$|\sC^{var}_t| - \sum_{k=1}^{\ell} (u_k+v_k) \geq \left(|\sC^{var}_t| - \frac{n}{\alpha} \right) \ge \left(\frac{1}{20}-\frac{1}{\alpha}\right)n$$ which is at least $n/40$ by choosing $\alpha = 40$. 
\end{proofof}

\medskip

In the rest of this section we prove Lemma~\ref{lem:main}.

\begin{cl} \label{cl:mat-tight-rank}
Given any $x\in P(\sM)$ with $\mathbf{0} <x<\mathbf{1}$, the maximum number of tight linearly independent rank constraints is  $n/2$. \end{cl}
\begin{proof} 
Recall that a tight constraint in $P(\sM)$ is any subset $T\sse V$ with $\sum_{i\in T} x_i = r(T)$. 
The  claim follows from the known property (see eg.~\cite{Sch}) that for any $x\in P(\sM)$ there is a linearly independent collection ${\cal C}$ of tight constraints such that (i) ${\cal C}$ spans all tight constraints and (ii) ${\cal C}$ forms a chain family. Since all right-hand-sides are integer and each variable is strictly between $0$ and  $1$, it follows that $|{\cal C}|\le \frac{n}2$. 
\end{proof}

\begin{cl}\label{cl:num-truc-steps}
The truncation Step~\ref{step:walk-trunc} occurs at most $n$ times.
\end{cl}
\begin{proof} We will show that whenever Step~\ref{step:walk-trunc} occurs (i.e. the random move gets truncated) the dimension $dim(\sV_{t+1})$ decreases by at least $1$, i.e.~$dim(\sV_{t+1})\le dim(\sV_t)-1$.
 As the maximum dimension is $n$ this would imply the claim.

Let $\sE_t$ denote the subspace spanned by all the tight constraints of $X_t\in \sQ$;
Recall that $\sV_t = \sE_t^{\perp}$ is the subspace orthogonal to $\sE_t$, and thus $dim(\sE_t)=n-dim(\sV_t)$. We also have $\sE_0\sse \sE_1\sse \cdots \sE_T$. Suppose that Step~\ref{step:walk-trunc} occurs in iteration $t$. Then we have $X_t\in \sQ$, $Y_t \not\in \sQ$ and $Y_t-X_t\in \sV_t$. Moreover  $X_{t+1} = X_t+\epsilon (Y_t-X_t) \in \sQ$ where $\epsilon \in [0,1)$ is such that $X_t+\epsilon' (Y_t-X_t) \not\in \sQ$ for all $\epsilon'>\epsilon$. 
So there is some constraint $\langle a', x\rangle \le \beta$ in $\sQ$ with:
$$\langle a', X_t\rangle \le \beta,\quad \langle a', X_{t+1}\rangle = \beta \quad \mbox{and} \quad \langle a', Y_t\rangle > \beta.$$

Since this constraint satisfies $\langle a', Y_t-X_t\rangle  > 0$ and $Y_t-X_t \in \sV_t$, we have $a'\not\in \sE_t$. As $a'$ is added to $\sE_{t+1}$, we have $dim(\sE_{t+1})\ge 1+dim(\sE_t)$. This proves the desired property and the claim.
\end{proof}

The statements of the following two lemmas are similar to those in \cite{LM12}, but the proofs require additional work since our random walk is different.  The first lemma shows that the expected number of tight side constraints at the end of the algorithm is not too high, and the second lemma shows that the expected number of tight variable constraints is large.

\begin{lemma}\label{lem:exp-deg-rank}
$\E[|\sC^{side}_T|]<\frac{n}4$.
\end{lemma}
\begin{proof}
Note that $X_T-y = \gamma\sum_{t=0}^T \oG_t + \sum_{q=1}^n \Delta_{t(q)}$ where $\Delta$s correspond to the truncation incurred during the iterations $t=t(1),\cdots,t(n)$ for which Step~\ref{step:walk-trunc} applies (by Claim~\ref{cl:num-truc-steps} there are at most $n$ such iterations). Moreover for each $q$, $\Delta_{t(q)} = \delta\cdot \oG_{t(q)}$ for some $\delta$ with $0<|\delta|< \gamma$.

If $j\in \sC_T^{side}$, then $|\langle X_T-y, a_j\rangle|=\lambda_j \sqrt{W_j(s)}+\frac1{n^3}\cdot \|a_j\|$. We have 
{\small $$
|\langle X_T-y, a_j\rangle| 
  \leq   |\gamma\sum_{t=0}^T \langle \oG_t, a_j\rangle| + \sum_{q=1}^n \gamma|\langle \oG_{a(q)}, a_j\rangle|  
  \leq |\gamma\sum_{t=0}^T \langle \oG_t, a_j\rangle| + n \gamma\cdot \max_{t=0}^T|\langle \oG_{t}, a_j\rangle|.
$$}
Note that at any iteration $t$,
$$|\langle \oG_{t}, a_j\rangle| = |\langle D^{-1}G_{t}, a_j\rangle|\le |\langle  G_{t}, a_j\rangle| \le \sum_{h=1}^k |\langle  b_h, a_j\rangle|\le n\|a_j\|.$$
The first inequality above uses that $D^{-1}$ is a diagonal matrix with entries at most one, the second inequality is by definition of $G_t$ where $\{b_h\}$ is an orthonormal basis of $\sV'_t$, and the last inequality uses that each $b_h$ is a unit vector. As $\gamma=n^{-6}$, we have $n \gamma\cdot \max_{t=0}^T|\langle \oG_{t}, a_j\rangle|\le \|a_j\|/n^4$. So it follows that if $j \in \sC_T^{side}$, then we must have:
$$|\gamma\sum_{t=0}^T \langle \oG_t, a_j\rangle| \,\, \ge \,\, \lambda_j \sqrt{W_j(s)}.$$

In order to bound the probability of this event, we consider the sequence $\{Z_t\}$ where $Z_t=\langle \oG_t, a_j\rangle$, and  note the following useful facts.  
\begin{observation}\label{obs:martingale}
The sequence $\{ Z_t\}$ forms a martingale satisfying:
\begin{enumerate}
\item $\E \left[ Z_{t} \mid Z_{t-1}, \ldots, Z_0 \right] = 0$ for all $t$. 
\item $|Z_t|\le n \|a_j\|$ whp for all $t$.
\item $\E \left[ Z_{t}^2 \mid Z_{t-1}, \ldots, Z_0 \right] \le \sum_{i=1}^n s_i^2\cdot a_{ji}^2=W_j(s)$ for all $t$. 
\end{enumerate}
\end{observation}
\begin{proof}
As $\oG_t=\sum_{h=1}^k g_h\cdot b_h$ where each $\E[g_h]=0$, we have $\E[\oG_t|\oG_0,\cdots,\oG_{t-1}]={\bf 0}$. Note that $\oG_t$ is not independent of  $\oG_0,\cdots,\oG_{t-1}$, as these choices determine the subspace where $\oG_t$ lies. So $\{Z_t\}$ forms a martingale sequence with the first property.

For the remaining two properties, we fix $j\in [m]$ and $t$ and condition on $Z_0,\cdots,Z_{t-1}$. To reduce notation we drop all subscripts: so $a=a_j$, $G=G_t$, $\sV'=\sV'_t$ and $Z=Z_t$.

Let $\{b_r\}$ denote an orthonormal basis for the linear subspace $\sV'$. Then $G=\sum_r g_r\cdot b_r$ where each $g_r$ is iid $\pm 1$ with probability half. 
As $\oG=D^{-1}G$, we have $Z=\langle \oG,a\rangle = \sum_r \langle D^{-1} b_r, a\rangle \, g_r = \sum_r \langle D^{-1}a,b_r\rangle \, g_r$. So, we can bound 
$$|Z| \,\, \le  \,\, \sum_r |\langle D^{-1}a, b_r\rangle| \cdot |g_r|  \,\, \le    \,\, \|D^{-1}a\| \sum_r |g_r|  \,\, \le  \,\, n\|a\|.  $$ 
The first inequality follows from the triangle inequality, the second by Cauchy-Schwartz and as $b_r$ is a unit-vector, and the third follows as $D^{-1}$ is a diagonal matrix with entries at most one. 
This proves property 2.

Finally, $\E[Z^2] = \sum_r \langle D^{-1}a,b_r\rangle^2 \, \E[g_r^2] = \sum_r \langle D^{-1}a,b_r\rangle^2 \le \|D^{-1}a\|^2$, where the last step follows as  $\{b_r\}$ is an orthonormal basis for a subspace of $\mathbb{R}^n$. This proves property 3. 
\end{proof}

Using a martingale  concentration inequality, we obtain:
\begin{cl}\label{cl:constr-violn-prob}
$\Pr\left[ |\gamma\sum_{t=0}^T \langle \oG_t, a_j\rangle| \,  \ge  \, \lambda_j \sqrt{W_j(s)} \right]  = \Pr\left[ | \sum_{t=0}^T Z_t| \,\ge \,  \frac{\lambda_j}{\gamma} \sqrt{W_j(s)} \right] \le 2\cdot \exp(-\lambda_j^2/3K)$.
\end{cl}
\begin{proof}  The first equality is by definition of the $Z_t$s. 
We now use the following concentration inequality:
  \begin{theorem}[Freedman~\cite{Freedman} (Theorem 1.6)]\label{thm:freedman}
    Consider a real-valued martingale sequence $\{Z_t\}_{t \geq 0}$ such
    that $Z_0 = 0$, $\E \left[ Z_{t} \mid Z_{t-1}, \ldots,
      Z_0 \right] = 0$ for all $t$, and $| Z_t | \leq M$ almost surely for all
    $t$. 
    Let $W_t = \sum_{j=0}^{t} \E \left[ Z^2_{j} \, \mid \, Z_{j-1}, Z_{j-2}, \ldots Z_0 \right]$
    for all $t \geq 1$. Then for all $\ell \geq 0$ and $\sigma^2 > 0$, and
    any stopping time $\tau$ we have
    \[
    \Pr \left[|\sum_{j=0}^\tau Z_{j}| \geq \ell \, \, and \, \, W_{\tau} \leq \sigma^2
    \right] \quad\leq \quad 2 \exp\left(-\frac{\ell^2/2}{\sigma^2 + M\ell/3} \right)\]
  \end{theorem}
   
We apply this with $M=n \|a_j\|$, $\ell=\frac{\lambda_j}{\gamma}\sqrt{W_j(s)}$, $\sigma^2=T\cdot W_j(s)$ and $\tau=T$. Note that 
$$\frac{\ell^2}{2\sigma^2+\frac23 M\ell} \quad = \quad \frac{\lambda_j^2 }{2\gamma^2T  + \frac23 \gamma n \|a_j\|\lambda_j/\sqrt{W_j(s)}}\quad \ge \quad \frac{\lambda_j^2 }{2\gamma^2T  + 1},$$
where the last inequality uses $W_j(s) \ge \|a_j\|^2/n^6$, $\lambda_j\le n$  and $\gamma = n^{-6}$. Thus
$$ \Pr\left[ |\gamma\sum_{t=0}^T \langle \oG_t, a_j\rangle| \,\, \ge \,\, \lambda_j \sqrt{W_j(s)} \right] \le 2\exp\left(\frac{-\lambda_j^2}{2\gamma^2T+1}\right)\le 2\cdot \exp(-\lambda_j^2/3K).$$
The last inequality uses $T=K/\gamma^2$ and $K\ge 1$. This completes the proof of the claim.
\end{proof}

By the above claim,   we have 
 $\E[|\sC^{side}_T|]< 2 \sum_{j=1}^m \exp(-\lambda_j^2/(30\alpha^2)) < 0.25n$. This completes the proof of Lemma~\ref{lem:exp-deg-rank}.
\end{proof}

\medskip

We now prove that in expectation, at least $0.1n$ variables become tight at the end of the algorithm. This immediately implies Lemma~\ref{lem:main}.
\begin{lemma} \label{lem:var-tight}
$\E[|\sC^{var}_T|]\ge 0.1n$.
\end{lemma}
\begin{proof}
Define the following potential function, which will measure the progress of the algorithm toward the variables becoming tight. 
$$\Phi(x) \quad :=\quad \sum_{k=1}^\ell 2^{2k}\cdot \left( \sum_{i\in U_k} x(i)^2 + \sum_{i\in V_k} (1-x(i)) ^2\right),\qquad \forall x\in \sQ.$$
Note that since $X_T\in \sQ$, we have $X_T(i)\le \alpha\cdot 2^{-k}$ for $i\in U_k$ and $1-X_T(i)\le \alpha\cdot 2^{-k}$ for $i\in V_k$. So $\Phi(X_T)\le \alpha^2\cdot n$.  
We also define the ``incremental function'' for any $x\in \sQ$ and $g\in \R^n$, $f(x,g) :=   \Phi(x+\gamma D^{-1} g)-\Phi(x)$. Recall that $D^{-1}$ is the $n\times n$ diagonal matrix with entries $(s_1,\cdots,s_n)$ where $s_i = 2^{-k}$ for $i\in U_k\cup V_k$. So
\begin{eqnarray*}
f(x,g) &= &\gamma^2 \sum_{i=1}^n g(i)^2 + 2\sum_{k=1}^\ell 2^{2k}\cdot \left( \sum_{i\in U_k} x(i)\gamma s_i\cdot g(i) - \sum_{i\in V_k} (1-x(i))\gamma s_i\cdot g(i)\right) \notag \\ 
&= &\gamma^2 \sum_{i=1}^n g(i)^2 + 2\gamma\sum_{k=1}^\ell  \left( \sum_{i\in U_k} \frac{x(i) g(i)}{s_i} - \sum_{i\in V_k} \frac{(1-x(i))g(i)}{ s_i}\right)
\end{eqnarray*}

Suppose the algorithm was modified to never have the truncation step~\ref{step:walk-trunc}, then in any iteration $t$, the increase $\Phi(Y_t)-\Phi(X_t)=f(X_t,G_t)$ where $G_t$ is the random direction chosen in $\sV'_t$. 
The following is by simple calculation.

{\small \begin{eqnarray}
f(X_t,G_t) - f(X_t,\delta G_t) & = & \gamma^2 (1-\delta^2) \|G_t\|_2^2 + 2\gamma(1-\delta) \sum_{k=1}^\ell \left( \sum_{i\in U_k} \frac{X_t(i)}{s_i}\cdot G_t(i) - \sum_{i\in V_k} \frac{1-X_t(i)}{s_i}\cdot G_t(i)\right) \notag \\
&\le & \gamma^2 (1-\delta^2) \|G_t\|_2^2 + 2\alpha\gamma(1-\delta) \sum_{i=1}^n |G_t(i)| \quad \le \quad \gamma^2\|G_t\|_2^2 + 2\gamma\alpha \|G_t\|_1 \notag\\
&\le & \gamma^2 n + 2\gamma\alpha n^{3/2} \quad \le \quad \frac1n \label{eq:potn-trunc}
\end{eqnarray}}
The first inequality in \eqref{eq:potn-trunc} uses the fact that $G_t$ is the sum of orthogonal unit vectors, and the second inequality uses $\gamma=n^{-6}$ and $\alpha=O(1)$.

This implies that  
\begin{eqnarray}   
\Phi(X_T)-\Phi(X_0) & = & \sum_{t=0}^T f(X_t,\delta_t G_t)   \,\ge\, \sum_{t=0}^T f(X_t, G_t) - \frac1n \sum_{t=0}^T \mathbf{1}[\mbox{step~\ref{step:walk-trunc} occurs in iteration }t] \nonumber \\
& \geq & \sum_{t=0}^T f(X_t, G_t) - 1 \label{eq:potn-tot} \qquad  \mbox{(by Claim~\ref{cl:num-truc-steps})}
\end{eqnarray}

\begin{cl}\label{cl:potn-inc}
$\E[\Phi(X_T)] -\Phi(y)  \ge \gamma^2 T\cdot \E[\dim(\sV_T)]-1$.
\end{cl}
\begin{proof} 
From~\eqref{eq:potn-tot} we have:
\begin{equation}\label{eq:exp-potn-inc}
\E[\Phi(X_T)]-\Phi(X_0) \ge \sum_{t=0}^T \E[f(X_t, G_t)] -1.
\end{equation}
In any iteration $t$, as $G_t=\sum_{h=1}^k g_h b_h$ where $\{b_h\}$ is an orthonormal basis for $\sV'_t$ and $g_h=\pm 1$,  
$$\E[f(X_t, G_t)] =  \gamma^2 \sum_{i=1}^n \E[G_t(i)^2] =  \gamma^2 \sum_{h=1}^k \|b_h\|^2 = \gamma^2 k = \gamma^2 \E[dim(\sV'_t)] = \gamma^2 \E[dim(\sV_t)].$$
Moreover, because $\sV_0\supseteq \sV_1 \supseteq \cdots \sV_T$, we have $\E[dim(\sV_t)]\ge  \E[dim(\sV_T)]$. So  
\begin{equation}\label{eq:iter-potn-inc}
\sum_{t=0}^T \E[f(X_t, G_t)]\quad \ge \quad \gamma^2 T\cdot \E[dim(\sV_T)].
\end{equation}

Combining \eqref{eq:exp-potn-inc} and \eqref{eq:iter-potn-inc},  we complete the proof of Claim~\ref{cl:potn-inc}. 
\end{proof}

By Claim~\ref{cl:mat-tight-rank} and the fact that $|\sC^{part}_T|=2\ell$, we have
\begin{eqnarray*}
\dim(\sV_T) & \ge &  n - \dim(\sC^{var}_T) - \dim(\sC^{side}_T) - \dim(\sC^{rank}_T)- \dim(\sC^{part}_T)\\
         & \ge &  \frac{n}{2}- 2\ell- \dim(\sC^{var}_T) - \dim(\sC^{side}_T) 
 \end{eqnarray*}
Taking expectations and  by Claim~\ref{lem:exp-deg-rank}, this gives
\begin{equation}
\label{eq:exp-dimvt}
 \E[\dim(\sV_T)] \ge  \frac{n}{4}- 2\ell- \dim(\sC^{var}_T) 
\end{equation}
Using $\Phi(X_T)\le \alpha^2 n$ and Claim~\ref{cl:potn-inc}, we obtain:
$$\alpha^2 n\ge \E[\Phi_T] \ge \gamma^2 T \cdot \left( \frac{n}{4}- 2\ell  - \E[\dim(\sC^{var}_T)] \right) - 1. $$
Rearranging and using $T= K/\gamma^2$, $K=10\alpha^2$ and $\ell=\log n$ gives that 
$$\E[\dim(\sC^{var}_T)]\ge \frac{n}{4}-\frac{\alpha^2 n}{K}-2\ell-\frac{1}{K} \ge 0.1n,$$
where we used $K=10\alpha^2$, $\alpha=40$ and  $\ell=O(\log n)$. This completes the  proof of  Lemma~\ref{lem:var-tight}.
\end{proof}


\section{Applications}

\subsection{Linear System Rounding with Violations}
Consider a $0-1$ integer program on $n$ variables where each constraint $j\in [m]$ corresponds to some subset $S_j\sse [n]$ of the variables having total value $b_j\in \mathbb{Z}_+$. That is,
$$P \quad = \quad \left\{x\in \{0,1\}^n \quad :\quad \sum_{i\in S_j} x_i = b_j,\,\, \forall j\in [m]\right\}.$$
\begin{theorem}\label{thm:IP}
There is a randomized polynomial time algorithm  that given any fractional solution 
satisfying the constraints in $P$, 
finds an integer solution $x\in \{0,1\}^n$ where for each $j\in[m]$, 
$$|x(S_j)-b_j|\quad \le \quad O(1)\cdot \min\left\{ \sqrt{j},\,\, \sqrt{n \log(m/n)}, \,\, \sqrt{\log m \log n \cdot b_j} + \log m \log n,\,\, \sqrt{\Delta}\log n \right\}.$$
Above $\Delta=\max_{i=1}^n |\{j\in [m] : i\in S_j\}|$ is the maximum number of constraints that any variable appears in. 
\end{theorem}
\begin{proof} Let $y\in [0,1]^n$ be a fractional solution with $\sum_{i\in S_j} y_i = b_j$  for all $j\in [m]$. The algorithm in Theorem~\ref{thm:IP} uses Theorem~\ref{thm:main} iteratively to obtain the integral solution $x$.

In each iteration, we start with a fractional solution $y'$ with $f\le n$ {\em fractional} variables  and set the parameters $\lambda_j$ suitably so that 
 $\sum_{j=1}^m e^{-\lambda_j^2/K_0} \le \frac{f}{16}$. That is, the condition in Theorem~\ref{thm:main} is satisfied.  
 Note that $W_j(y')=\sum_{i\in S_j} (y'_i)^2\le y'(S_j)$ and $W_j(y')\le f$. Now, by applying Theorem~\ref{thm:main}, we would obtain a new fractional solution $y''$ such that: 
\begin{itemize}
\item For each $j\in [m]$, $|y''(S_j)-y'(S_j)|\le \lambda_j \sqrt{W_j(y')} + \frac1n \le O(\lambda_j)\cdot \sqrt{f}$.
\item The number of fractional variables in $y''$ is at most $\frac{f}{K}$ for some constant $K>1$.
\end{itemize}  
Therefore, after $\frac{\log n}{\log K} = O(\log n)$ iterations we obtain a solution with $O(1)$ fractional variables. Setting these fractional variables arbitrarily to $0-1$ values, we obtain an integral solution $x$.

Let us partition the constraints into sets $M_1, M_2, M_3$ and $M_4$ based on which of the four terms in Theorem~\ref{thm:IP} is minimum. That is, $M_1\sse [m]$ consists of constraints $j\in [m]$ where $\sqrt{j}$ is smaller than the other three terms; $M_2,M_3,M_4$ are defined similarly. Below we show how to set the parameters $\lambda_j$ and bound the constraint violations for these parts separately.

\paragraph{Error bound of $\min\{\sqrt{j},\, \sqrt{n\log(m/n)}\}$ for $j\in M_1\cup M_2$.} In any iteration with $f\le n$ fractional variables, we set the parameters $\lambda_j$s in Theorem~\ref{thm:main} as follows:
$$\lambda_j = \left\{
\begin{array}{ll}
0 & \mbox{ if }j<c_1 f\\
\sqrt{c_2\, \log\frac{j}{c_1f}} & \mbox{ if }j\ge c_1 f
\end{array}\right.
$$
Here $c_1$ and $c_2$ are constants that will be fixed later. 
Note that  
$$\sum_{j\in M_1\cup M_2}^m e^{-\lambda_j^2/K_0} \, \le \, c_1 f + \sum_{j\ge c_1 f} e^{-\frac{c_2}{K_0}\log\frac{j}{c_1f}} \,\le \, c_1f + \sum_{i\ge 0} 2^ic_1f \cdot e^{-ic_2/K_0} \, \le \, c_1f + c_1 f \sum_{i\ge 0} 2^{-i} \, \le \, 3c_1f,$$
which is at most $f/48$ for $c_1<1/150$. The second inequality above is obtained by bucketing the $j$s into intervals of the form $[2^i\cdot c_1f,\, 2^{i+1}\cdot c_1f]$. The third inequality uses $c_2\ge 2K_0$.

We now bound the error incurred. 
\begin{enumerate}
\item Consider first a constraint $j\le n$. Note that $\lambda_j$ stays zero until the number of fractional variables $f$ drops below $j/c_1$. So we can bound $|x(S_j)-b_j|$ by:
$$\sum_{i\ge 0} \sqrt{c_2 \frac{j}{c_1K^i}\cdot \log K^i} \le O(\sqrt{j}) \, \sum_{i\ge 0} \sqrt{i}K^{-i/2} = O(\sqrt{j}),$$
where $i$ indexes the iterations of the algorithm after $f$ drops below $j/c_1$ for the first time.

\item Now consider a constraint $j>n$. Similarly, we bound $|x(S_j)-b_j|$ by:
$$\sum_{i\ge 0} \sqrt{c_2 \frac{n}{K^i}\cdot \log(\frac{j}{c_1n} K^i)} \le O(\sqrt{n\log(j/n)}) \, \sum_{i\ge 0} \sqrt{i}K^{-i/2} = O(\sqrt{n\log(j/n)}).$$
Here $i$ indexes the number of iterations of the algorithm from its start.
\end{enumerate}

\paragraph{Error bound of $\sqrt{L\cdot b_j} + L$ for $j\in M_3$, where $L=\Theta(\log m \log n)$.} Note that the additive term in this expression is at least $L$. If any $b_j<L$ then we increase it to $L$ (and add dummy elements to $S_j$ and ensure $y(S_j)=L$); this only affects the error term by a constant factor as $L\le \sqrt{L\cdot b_j} + L\le 2L$. So in the following we assume that $\min_j b_j \ge L$.

Here we set $\lambda_j=\infty$ in all iterations, which satisfies $\sum_{j\in M_3} e^{-\lambda_j^2/K_0} =0$.

The analysis of the error incurred is similar to that in Lemma~\ref{lem:exp-deg-rank} and we only sketch the details;  the main difference is that we analyze the deviation in a combined manner over all $O(\log n)$ iterations. Fix any constraint $j\in [m]$. If we ignore the error due to the truncation steps over all iterations\footnote{This can be bounded by $o(1)$ exactly as in Event 2 of Lemma~\ref{lem:exp-deg-rank}.}  then we can write $|x(S_j)-b_j| = |\sum_{t=0}^P \gamma Z_t|$ where $\gamma=n^{-6}$ and $Z_t=\langle \oG_t , \mathbf{1}_{S_j}\rangle$; recall that each $\oG_t=D^{-1}G_t$ for random direction $G_t$ as in  Step~\ref{step:gaussian} of the algorithm in Section~\ref{sec:main-rounding}. Here $P= O(\log n/\gamma^2)$ since there are $O(\log n)$ iterations and $O(1/\gamma^2)$ steps in each iteration. We will use the concentration inequality in Theorem~\ref{thm:freedman} with martingale $\{Z_t\}_{t\ge 0}$ and stopping time $\tau$ being the first time $t'$ where $|\sum_{t=0}^{t'} Z_t| > \frac{1}{\gamma}\sqrt{Lb_j}$. Then it follows that at any step $t'$ before stopping, the current solution $y'$ satisfies $y'(S_j)-y(S_j) = \gamma \sum_{t=0}^{t'} Z_t \le \sqrt{Lb_j} \le b_j$ (using the  assumption $b_j\ge L$), i.e. $y'(S_j)\le 2b_j$. Now we can bound $W_\tau \le P \cdot O(b_j) = O(\log n/\gamma^2)\cdot b_j$. Using Theorem~\ref{thm:freedman} with $\ell=\sqrt{Lb_j}/\gamma$, we obtain:
$$ \Pr\left[ |\gamma\sum_{t=0}^\tau Z_t | \,\, \ge \,\,  \sqrt{Lb_j} \right] \le 2\exp\left(\frac{-Lb_j}{O(\log n)b_j}\right)\le \frac{1}{m^2},$$
by choosing a large enough constant in $L=O(\log m \log n)$. It follows that with probability at least $1-m^{-2}$, we have $\tau=P$ and 
$|x(S_j)-b_j|= |\sum_{t=0}^P \gamma Z_t| \le \sqrt{L\cdot b_j}$.
Finally, taking a union bound over $|M_3|\le m$ such events, we obtain that with high probability, $|x(S_j)-b_j|\le \sqrt{L\cdot b_j}$ for all $j\in M_3$.

\paragraph{Error bound of $\sqrt{\Delta}\log n$ for $j\in M_4$.}
Here we set  $\lambda_j= \sqrt{K_1\Delta}/\sqrt{|S_j|}$ in all iterations, where  $K_1$ is a constant to be fixed later. We first bound $\sum_{j\in M_4} e^{-\lambda_j^2/K_0}$. Note that when restricted to the $f$ fractional variables in any iteration, $\sum_{j=1}^m |S_j|\le \Delta f$ since each variable appears in at most $\Delta$ constraints. So the number of constraints with $|S_j|>64\Delta$ is at most $\frac{f}{64}$. For $h\ge 0$, the number of constraints with $|S_j| \in [2^{-h-1}64\Delta, 2^{-h}64\Delta)$ is at most $2^{h+1}\frac{f}{64}$. 
So,
$$\sum_{j\in M_4} e^{-\lambda_j^2/K_0} \le \frac{f}{64} + \sum_{h=0}^\infty 2^{h+1}\frac{f}{64} \exp\left(\frac{-K_1 \Delta}{2^{-h}64\Delta\cdot K_0}\right)\le \frac{f}{64} + \frac{f}{64}\sum_{h=0}^\infty 2^{h+1}  e^{-2^{h+2}}\le \frac{f}{48}.$$
The second inequality is by choosing large enough constant $K_1$.

We now bound the error incurred for any constraint $j\in M_4$. The error in a single iteration is at most $O(\sqrt{\Delta}) + \frac{1}{n}$. So the overall error $|x(S_j)-b_j|  = O(\sqrt{\Delta}\log n)$.

\paragraph{Overall iteration.} By setting the $\lambda_j$ parameters for the different parts $M_1,M_2,M_3,M_4$ as above, it follows that in any iteration with $f$ fractional variables, we have 
$\sum_{j=1}^m e^{-\lambda_j^2/K_0} \le \frac{f}{24}$ which satisfies the condition in Theorem~\ref{thm:main}. 
\end{proof}

\medskip

\noindent {\em Remark:} The above result also extends to the following ``group sparse'' setting. Suppose the constraints in $M_4$ are further partitioned into $g$ groups $\{G_k\}_{k=1}^g$ where the column sparsity restricted to constraints in each group $G_k$ is $\Delta_k$. Then we obtain an integral solution with $|x(S_j)-b_j|= O(\sqrt{g\cdot \Delta_k}\, \log n)$ for all $j\in G_k$. The only modification required in the above proof is to set $\lambda_j=\sqrt{K_1\cdot g\cdot \Delta_k}/\sqrt{|S_j|}$ for $j\in G_k$.

\subsection{Minimum Cost Degree Bounded Matroid Basis}

The input to the {\em minimum cost degree bounded matroid} problem (\dbm) is a matroid defined on elements $V=[n]$ with costs $d:V\rightarrow \mathbb{Z}_+$ and $m$ ``degree constraints'' $\{S_j, b_j\}_{j=1}^m$ where each $S_j\sse [n]$ and $b_j\in \mathbb{Z}_+$. 
The objective is to find a minimum-cost base $I$ in the matroid that obeys all the degree bounds, i.e. $|I\cap S_j|\le b_j$ for all $j\in[m]$. Here we make a minor technical assumption that all costs are polynomially bounded integers.

An algorithm for \dbm is said to be an $(\alpha, \beta\cdot b + \gamma)$-bicriteria approximation algorithm if for any instance, it finds 
a base $I$ satisfying $|I\cap S_j|\le \beta\cdot b_j+\gamma$ for all $j\in[m]$ and having cost at most $\alpha$ times the optimum (which satisfies all degree bounds).

\begin{theorem}\label{thm:dbm}
There is a randomized algorithm for \dbm, that on any instance, finds a base $I^*$ of cost at most the optimum where for each $j\in [m]$:
$$|I^*\cap S_j|\quad \le \quad O(1)\cdot \min\left\{ \sqrt{j},\,\, \sqrt{n \log(m/n)}, \,\, \sqrt{\log m \log n \cdot b_j} + \log m \log n,\,\, \sqrt{\Delta}\log n \right\}.$$
\end{theorem}
\begin{proof}
Let $y\in [0,1]^n$ be an optimal solution to the natural LP relaxation of \dbm. We now describe the rounding algorithm: this is based on iterative applications of Theorem~\ref{thm:mat-disc}. First, we incorporate the cost as a special degree constraint $v_0= d$ indexed zero. 
We will require zero violation in the cost during each iteration, i.e. $\lambda_0=0$ always. We partition the degree constraints $[m]$ as in Theorem~\ref{thm:IP}: recall the definitions of $M_1,M_2,M_3,M_4$, and the setting of their $\lambda_j$ parameters in each iteration.

In each iteration, we start with a fractional solution $y'$ with $f\le n$ {\em  fractional} variables. Using the same calculations as Theorem~\ref{thm:IP}, we have $\sum_{j=0}^m e^{-\lambda_j^2/K_0} \le 1+\frac{f}{24}\le \frac{f}{16}$ assuming $f\ge 48$. 
For now assume $f\ge \max\{K_0, 48\}$; applying Theorem~\ref{thm:mat-disc}, we obtain a new fractional solution $y''$ that has:
\begin{itemize}
\item $|\langle v_0 , y'' -y' \rangle |\le \|d\|/n^{O(1)}\le \frac1n$.
\item For each $j\in [m]$,  $|y''(S_j)-y'(S_j)|\le \lambda_j \sqrt{W_j(y')} + \frac1n$. 
\item The number of  fractional variables in $y''$ is at most $\frac{f}{K'}$ for some constant $K'>1$.
\end{itemize}  
The first condition uses the fact that the error term $\|a_j\|/n^2$ in Theorem~\ref{thm:mat-disc} can be reduced to $\|a_j\|/n^c$ for any constant $c$, and that $\|d\|\le poly(n)$ as we assumed all costs to be polynomially bounded.

We repeat these iterations as long as $f\ge \max\{K_0,48\}$ : this takes $T\le \frac{\log n}{\log K'} = O(\log n)$ iterations. The violation in the cost (i.e. constraint $j=0$) is at most $\frac{T}{n}<1$. For any degree constraint $j\in [m]$, the violation is exactly as in Theorem~\ref{thm:IP}.

At the end of the above iterations, we are left with an almost integral solution $x$: it has $O(1)$ fractional variables. Notice that $x$ lies in the matroid base polytope: so it can be expressed as a convex combination of (integral) matroid bases. We output the minimum cost base $I^*$ in this convex decomposition of $x$. Note that the cost of solution $I^*$ is at most that of $x$ which is less than  $\langle d, y\rangle + 1$.  Moreover, $I^*$ agrees with $x$ on all integral variables of $x$: so the worst case additional violation of any degree constraint is just $O(1)$.
\end{proof}

We state two special cases of this result, which improve on prior work.
\begin{corollary}\label{cor:dbm}
There are randomized bicriteria approximation algorithms for \dbm with ratios $(1,b+O(\sqrt{n\log(m/n)}))$ and $(1, O(\sqrt{\Delta}\log n))$. 
\end{corollary}

Previously, \cite{CVZ10} obtained a $(1,b+O(\sqrt{n\log(m)}))$ bicriteria approximation and~\cite{KLS08} obtained a $(1,\Delta-1)$ bicriteria approximation for \dbm.

\def\mcm{\ensuremath{\mathsf{MCM}}\xspace}
\subsection{Multi-criteria Matroid Basis}
The input to the {\em multi-criteria matroid basis} is a matroid \sM defined on elements $V=[n]$ with $k$ different cost functions $d_j:[n]\rightarrow \mathbb{Z}_+$ (for $j=1,\cdots,k$) and budgets $\{B_j\}_{j=1}^k$. The goal is to find (if possible) a basis  $I$ with $d_j(I)\le B_j$ for each $j\in [k]$. We obtain:
\begin{theorem}\label{thm:mcm}
There is a randomized algorithm for multi-criteria matroid basis, that given any $\epsilon>0$ finds in $n^{O(k^{1.5}\,/\,\epsilon)}$ time, a basis $I$ with $d_j(I)\le (1+\epsilon)B_j$ for all $j\in[k]$. 
\end{theorem}
Previously, \cite{GRSZ14} obtained a deterministic algorithm for \mcm that required $n^{O(k^2\,/\,\epsilon)}$ time. One could also use the algorithm of \cite{CVZ10} to obtain a randomized PTAS for \mcm, but this approach requires at least $n^{\Omega(k\,/\,\epsilon^2)}$ time. Our running time is better when $\epsilon<1/\sqrt{k}$.

We now describe the algorithm in Theorem~\ref{thm:mcm}. An element $e$ is said to be {\em heavy} if its $j^{th}$ cost $d_j(e)>\frac{\epsilon}{\sqrt{k}} B_j$ for any $j\in[k]$. Note that the optimal solution contains at most $\frac{k^{1.5}}{\epsilon}$ heavy elements. The algorithm first guesses by enumeration all heavy elements in the optimal solution. Let $\sM'$ denote the matroid obtained by contracting these heavy elements. Let $B'_j$ denote the residual budget for each $j\in[k]$.  The algorithm now solves the natural LP relaxation:
$$x\in P(\sM'),\quad \langle d_j, x\rangle \le B'_j,\,\, \forall j\in[k].$$

The rounding algorithm is an iterative application of Theorem~\ref{thm:mat-disc}: the number of fractional variables decreases by a factor of $K>1$ in each iteration. 

As long as the 
number of fractional variables $n'<16k$, we use $\lambda_j=0$ for all $j\in[k]$; note that this satisfies the condition  $\sum_{j=1}^k e^{-\lambda_j^2/K_0} \le n'/16$. Note that there is no loss in any of the budget constraints in this first phase of the rounding. 

Once $n'\le N:=16k$, we choose each $\lambda_j=\sqrt{K_0\log (N/n')}$ which satisfies the condition on $\lambda$s. The loss in the $j^{th}$ budget constraint in such an iteration is at most $\lambda_j\sqrt{n'}\cdot d_j^{max}$ where $d_j^{max}\le \frac{\epsilon}{\sqrt{k}} B_j$ is the maximum cost of any element. So the increase in the  $j^{th}$ budget constraint over all iterations is at most:

$$d_j^{max} \cdot \sum_{i=0}^{t-1} \sqrt{K_0\,\frac{N}{K^i}\, \log (K^i)} \,\, \le  \,\, O(\sqrt{N})\cdot d_j^{max}  \,\, =  \,\, O(\epsilon)B_j.$$
Above $i$ indexes iterations in the second phase of rounding.

\subsection{Low Congestion Routing on Short Paths}
The routing on short paths (\rsp) problem is defined on an $n$-vertex directed graph $G=(V,E)$ with edge capacities $b:E\rightarrow \mathbb{Z}_+$. There are
$k$ source-sink pairs $\{(s_i,t_i)\}_{i=1}^k$ and a length bound $\Delta $. The goal in \rsp is to find an $s_i-t_i$ path $P_i$ of length at most $\Delta $ for each pair $i\in[k]$ such that the number of paths using any edge $e$ is at most $b_e$. 

The decision problem of determining whether there exist such paths is NP-complete. Hence we focus on bicriteria approximation algorithms, where we attempt to find paths $P_i$s that violate the edge capacities by a small amount.  As noted in~\cite{CVZ10}, we can use any LP-based  algorithm for \dbm to obtain one for \rsp: for completeness we describe this briefly below.

Let $\p_i$ denote the set of all $s_i-t_i$ paths of length at most $\Delta$. Consider the following LP relaxation for \rsp.
\begin{eqnarray*}
\sum_{P\in \p_i} x_{i,P} &\ge &1,\qquad \forall i\in [k]\\
\sum_{i=1}^k \sum_{P\in \p_i: e\in P} x_{i,P} &\le &b_e,\qquad \forall e\in E\\
x&\ge &0.
\end{eqnarray*}
Although this LP has an exponential  number of variables, it can be solved in polynomial time by an equivalent polynomial-size formulation using  a ``time-expanded network''. 

Given any feasible instance of \rsp, we obtain a fractional solution to the above LP. Moreover, the number of non-zero variables $x_{i,P}$ is at most $k+|E|=poly(n)$. Let $\p'_i$ denote the set of $s_i-t_i$ paths with non-zero value in this fractional solution. Consider now an instance of \dbm on groundset $U=\cup_{i=1}^k \p'_i$ where the matroid is a {\em partition matroid} that requires one element from each $\p'_i$. The degree constraints correspond to edges $e\in E$, i.e. $S_e=\{P\in U : e\in P\}$. The goal is to find a base $I$ in the partition matroid such that $|S_e\cap I|\le b_e$ for all $e\in E$. Note that the column sparsity of the degree constraints is $\Delta$ since each path in $U$ has length at most $\Delta$. Moreover $\{x_{i,P}\, :\, P\in \p'_i,\, i\in [k]\}$ is a feasible fractional solution to the LP relaxation of this \dbm instance. So we obtain:
\begin{corollary}\label{cor:rsp}
There is an algorithm that given any feasible instance of \rsp, computes an $s_i-t_i$ path of length at most $\Delta $ for each $i\in [k]$ where the number of paths using any edge $e$ is at most $b_e+\min\left\{O(\sqrt{\Delta }\log n),\, O(\sqrt{b_e} \log n+\log^2n)\right\}$. 
\end{corollary}

\paragraph{Multipath routing with laminar requirements} Our techniques can also handle a richer set of requirements in the \rsp problem. In addition to the graph $G$, pairs $\{(s_i,t_i)\}_{i=1}^k$ and length bound $\Delta$, there is a laminar family ${\cal L}$ defined on the pairs $[k]$ with an integer requirement $r_T$ on each set $T\in {\cal L}$. The goal in the {\em laminar \rsp} problem is to find a multiset of $s_i-t_i$ paths (for $i\in[k]$) such that:
\begin{enumerate}
\item each path has length at most $\Delta$,
\item for each $T\in {\cal L}$, there are at least $r_T$ paths between pairs of $T$, and
\item the number of paths using any edge $e$ is at most $b_e$. 
\end{enumerate}

Consider the following LP relaxation for this problem.
\begin{eqnarray*}
\sum_{i\in T} \,\, \sum_{P\in \p_i} x_{i,P} &\ge &r_T,\qquad \forall T\in {\cal L}\\
\sum_{i=1}^k \,\, \sum_{P\in \p_i: e\in P} x_{i,P} &\le &b_e,\qquad \forall e\in E\\
x&\ge &0.
\end{eqnarray*}
This LP can again be solved using an equivalent polynomial-sized LP. Let $\p'_i$ denote the set of $s_i-t_i$ paths with non-zero value in this fractional solution, and define groundset $U=\cup_{i=1}^k \p'_i$. As before, we also define ``degree constraints'' corresponding to edges $e\in E$, i.e. at most $b_e$ elements can be chosen from $S_e=\{P\in U : e\in P\}$.
Unlike the usual \rsp problem we can not directly cast these laminar requirements as a matroid constraint, 
but a slight modification of the \dbm algorithm works. 

The main idea is that the partial rounding result (Theorem~\ref{thm:mat-disc}) also holds if we want to exactly preserve any laminar family ${\cal L}$ of constraints (instead of a matroid). Note that a laminar family on $|U|$ elements might have $2|U|$ sets. However, it is easy to see that the number of {\em tight} constraints of ${\cal L}$ at any strictly fractional solution is at most $|U|/2$. Using this observation in place of Claim~\ref{cl:mat-tight-rank}, we obtain the partial rounding result also for laminar constraints.

Finally using this partial rounding as in Theorem~\ref{thm:dbm}, we obtain:
\begin{theorem}\label{thm:lam-rsp}
There is an algorithm that given any feasible instance of laminar \rsp, computes a multiset ${\cal Q}$ of  $s_i-t_i$ paths such that:
\begin{enumerate}
\item each path in ${\cal Q}$ has length at most $\Delta$,
\item for each $T\in {\cal L}$, there are at least $r_T$ paths in ${\cal Q}$ between pairs of $T$, and
\item the number of paths in ${\cal Q}$ using any edge $e$ is at most: 
$$b_e+\min\left\{O(\sqrt{\Delta }\log n),\, O(\sqrt{b_e} \log n+\log^2n)\right\}.$$ 
\end{enumerate}
\end{theorem}

\bibliographystyle{alpha}
\bibliography{mybib}

\newcommand{\etalchar}[1]{$^{#1}$}
\begin{thebibliography}{BKK{\etalchar{+}}13}

\bibitem[Ban10]{B10}
Nikhil Bansal.
\newblock Constructive algorithms for discrepancy minimization.
\newblock In {\em Foundations of Computer Science (FOCS)}, pages 3--10, 2010.

\bibitem[BCKL14]{BCKL14}
Nikhil Bansal, Moses Charikar, Ravishankar Krishnaswamy, and Shi Li.
\newblock Better algorithms and hardness for broadcast scheduling via a
  discrepancy approach.
\newblock In {\em {SODA}}, pages 55--71, 2014.

\bibitem[BF81]{BF81}
J.~Beck and T.~Fiala.
\newblock Integer-making theorems.
\newblock {\em Discrete Applied Mathematics}, 3:1--8, 1981.

\bibitem[BKK{\etalchar{+}}13]{BKKNP13}
Nikhil Bansal, Rohit Khandekar, Jochen K{\"{o}}nemann, Viswanath Nagarajan, and
  Britta Peis.
\newblock On generalizations of network design problems with degree bounds.
\newblock {\em Math. Program.}, 141(1-2):479--506, 2013.

\bibitem[CNN11]{CNN11}
Moses Charikar, Alantha Newman, and Aleksandar Nikolov.
\newblock Tight hardness results for minimizing discrepancy.
\newblock In {\em SODA}, pages 1607--1614, 2011.

\bibitem[CVZ10]{CVZ10}
Chandra Chekuri, Jan Vondrak, and Rico Zenklusen.
\newblock Dependent randomized rounding via exchange properties of
  combinatorial structures.
\newblock In {\em {FOCS}}, pages 575--584, 2010.

\bibitem[Fre75]{Freedman}
David~A. Freedman.
\newblock On tail probabilities for martingales.
\newblock {\em Annals of Probability}, 3:100--118, 1975.

\bibitem[GRSZ14]{GRSZ14}
Fabrizio Grandoni, R.~Ravi, Mohit Singh, and Rico Zenklusen.
\newblock New approaches to multi-objective optimization.
\newblock {\em Math. Program.}, 146(1-2):525--554, 2014.

\bibitem[HSS14]{HSS14}
Nicholas J.~A. Harvey, Roy Schwartz, and Mohit Singh.
\newblock Discrepancy without partial colorings.
\newblock In {\em {APPROX/RANDOM} 2014}, pages 258--273, 2014.

\bibitem[KLR{\etalchar{+}}87]{KLRTVV87}
Richard~M. Karp, Frank~Thomson Leighton, Ronald~L. Rivest, Clark~D. Thompson,
  Umesh~V. Vazirani, and Vijay~V. Vazirani.
\newblock Global wire routing in two-dimensional arrays.
\newblock {\em Algorithmica}, 2:113--129, 1987.

\bibitem[KLS08]{KLS08}
Tam{\'{a}}s Kir{\'{a}}ly, Lap~Chi Lau, and Mohit Singh.
\newblock Degree bounded matroids and submodular flows.
\newblock In {\em IPCO}, pages 259--272, 2008.

\bibitem[LLRS01]{LLRS01}
Frank~Thomson Leighton, Chi{-}Jen Lu, Satish Rao, and Aravind Srinivasan.
\newblock New algorithmic aspects of the local lemma with applications to
  routing and partitioning.
\newblock {\em {SIAM} J. Comput.}, 31(2):626--641, 2001.

\bibitem[LM12]{LM12}
Shachar Lovett and Raghu Meka.
\newblock Constructive discrepancy minimization by walking on the edges.
\newblock In {\em FOCS}, pages 61--67, 2012.

\bibitem[LRS11]{Lau-book}
Lap-Chi Lau, R.~Ravi, and Mohit Singh.
\newblock {\em Iterative Methods in Combinatorial Optimization}.
\newblock Cambridge University Press, 2011.

\bibitem[LSV86]{LSV86}
L.~Lovasz, J.~Spencer, and K.~Vesztergombi.
\newblock Discrepancy of set-systems and matrices.
\newblock {\em European J. Combin.}, 7:151--160, 1986.

\bibitem[Mat10]{Matousek2010geometric}
J.~Matou\v{s}ek.
\newblock {\em Geometric Discrepancy: An Illustrated Guide}.
\newblock Springer, 2010.

\bibitem[NT15]{NT15}
Aleksandar Nikolov and Kunal Talwar.
\newblock Approximating hereditary discrepancy via small width ellipsoids.
\newblock In {\em Symposium on Discrete Algorithms, {SODA}}, pages 324--336,
  2015.

\bibitem[Rot13]{R13}
Thomas Rothvoss.
\newblock Approximating bin packing within o(log {OPT} * log log {OPT)} bins.
\newblock In {\em {FOCS}}, pages 20--29, 2013.

\bibitem[Rot14]{R-convex14}
Thomas Rothvoss.
\newblock Constructive discrepancy minimization for convex sets.
\newblock In {\em {IEEE} Symposium on Foundations of Computer Science, {FOCS}},
  pages 140--145, 2014.

\bibitem[Sch03]{Sch}
A.~Schrijver.
\newblock {\em {C}ombinatorial {O}ptimization}.
\newblock Springer, 2003.

\bibitem[SL07]{SL07}
Mohit Singh and Lap~Chi Lau.
\newblock Approximating minimum bounded degree spanning trees to within one of
  optimal.
\newblock In {\em {STOC}}, pages 661--670, 2007.

\bibitem[Spe85]{Spencer85}
Joel Spencer.
\newblock Six standard deviations suffice.
\newblock {\em Transactions of the American Mathematical Society},
  289(2):679--706, 1985.

\bibitem[Sri97]{Srin97}
Aravind Srinivasan.
\newblock Improving the discrepancy bound for sparse matrices: Better
  approximations for sparse lattice approximation problems.
\newblock In {\em Symposium on Discrete Algorithms (SODA)}, pages 692--701,
  1997.

\bibitem[Vaz01]{V01}
Vijay~V. Vazirani.
\newblock {\em Approximation Algorithms}.
\newblock Springer-Verlag, 2001.

\bibitem[WS11]{SW11}
David Williamson and David Shmoys.
\newblock {\em The design of Approximation Algorithms}.
\newblock Cambridge University Press, 2011.

\end{thebibliography}

\appendix

\section{Useful Linear Programming Facts}

\begin{fact}\label{fact:LP-point}
Consider any polyhedron given by $P=\{x: Ax\le b\}$ where all entries in $A,b$ are integers of size at most  $\log_2 B$. Then there is a polynomial (in $\log  B$ and size of $A$) time algorithm that  given any point $u\in P$,  
finds another point $v^*\in P$ where (i) $\|u-v^*\|_1 \le \frac{1}{n^7 B}$ and (ii) all entries in $v^*$ are rationals of size $O(n^2\log B)$.
\end{fact}
\begin{proof} Let $L:=2n^8 B$ and  $u'$ denote the point with coordinates $u'_i = \frac{1}{L} \lfloor L\cdot u_i\rfloor$ for all $i\in[n]$. 
We now write a linear program that computes the point $v\in P$ with minimum $\ell_1$ distance from $u'$. 
$$\begin{array}{lll}
\min &\sum_{i=1}^n d_i &\\
\mbox{s.t.} &  Av\le b & \\
& |Lv_i-Lu'_i|\le Ld_i & \forall i\in [n]\\
&\sum_{i=1}^n d_i \le 1 &  \\
& \mathbf{d}, \mathbf{v} \in \R^n. &
\end{array}
$$
Note that the feasible region of this LP is a polytope (bounded polyhedron) due to the last two constraints. So there is an  optimal extreme point solution $v^*$ that can be found in polynomial time. Since   
all constraint coefficients in this LP are integers bounded  by $L$, the entries in $v^*$ must be rationals bounded  by $(2nL)^{2n}$. Finally,    $u\in P$ corresponds to a feasible solution to this LP with $v=u$, $d_i = |v_i-u_i|$ (for $i\in [n]$) and objective $\|u-u'\|_1\le \frac{n}{L}$. It now follows that $\|v^*-u'\|_1 \le \frac{n}{L}$ and so $\|v^*-u\|_1 \le \|v^*-u'\|_1 + \|u'-u\|_1 \le \frac{2n}{L}$. 
\end{proof}

\begin{fact}\label{fact:gauss-elim}
Consider any linear subspace given by $\{ x : Ax=0\}$ where all entries in $A$ are integers of size at most  $\log_2 B$. Then there is a polynomial (in $\log  B$ and size of $A$) time algorithm that computes a basis $\{b_j\}_{j=1}^k$ of this subspace where (i) all entries are rationals of size $O(n^2  \log B)$, (ii) $|\langle b_j,b_j\rangle -1 | \le \frac{1}{n^4B}$ for all $j\in [k]$, and (iii) $|\langle b_j,b_\ell\rangle  | \le \frac{1}{n^4B}$ for all $j\ne \ell$, $j, \ell \in [k]$.
\end{fact}
\begin{proof}
We can obtain an orthonormal basis $\{b'_j\}_{j=1}^k$ of this subspace using Gaussian elimination and Gram-Schmidt orthogonalization.  This clearly satisfies the last two conditions. But some more work is needed since we require the entries in the basis vectors to be bounded integers. 

To ensure this, we modify each vector $b'_j$  into $b_j$ separately 
by applying Fact~\ref{fact:LP-point} with polyhedron $P=\{x:Ax=0\}$, $u=b'_j$, and then set $b_j=v^*$. Now the last condition follows from Fact~\ref{fact:LP-point}(ii). The first and second conditions follow from Fact~\ref{fact:LP-point}(i) since $\{b'_j\}_{j=1}^k$ is orthonormal. 
\end{proof}

\end{document}